\documentclass[journal]{IEEEtran}
\usepackage{graphicx}
\usepackage{epstopdf}
\usepackage{amsmath, amssymb, amsthm, amsfonts}
\usepackage{ifthen}
\usepackage{tikz}
\usetikzlibrary{arrows,shapes}
\usepackage{enumerate}
\usepackage{verbatim}
\usepackage{float}
\usepackage{pgfplots}
\pgfplotsset{compat=1.15}

\def\eE{\mathbb E}

\def\bhcY{\mathcal{Y}^n}
\def\bhcS{\mathcal{S}^n}
\def\sCapa{\mathcal C_{\text{est}}^{\text{fad}}(\mathrm{P})}

\newtheorem{theorem}{Theorem}

\newtheorem{lemma}{Lemma}

\newtheorem{defn}{Definition}

\DeclareMathOperator{\Var}{Var}

\allowdisplaybreaks
\DeclareSymbolFont{bbold}{U}{bbold}{m}{n}
\DeclareSymbolFontAlphabet{\mathbbold}{bbold}

\def\eE{\mathbb E}

\def\bS{\mathbf S}

\def\bY{\mathbf Y}
\def\bhS{\hat{\mathbf S}}

\def\bZ{\mathbf Z}

\def\bhS{\hat{\mathbf S}}

\def\btS{\widetilde{\mathbf S}}

\def\bX{\mathbf X}

\def\bhcY{\mathcal{Y}^n}
\def\bhcS{\mathcal{S}^n}

\def\va{\boldsymbol{\alpha}}

\def\sCapa{\mathcal C_{\text{CR}}^{\text{fad}}(\mathrm{P})}

\newcommand{\be}{\begin{equation}}
\newcommand{\ee}{\end{equation}}
\newcommand{\ben}{\begin{equation*}}
\newcommand{\een}{\end{equation*}}
\newcommand{\ba}{\begin{eqnarray}}
\newcommand{\ea}{\end{eqnarray}}

\allowdisplaybreaks
\begin{document}
\title{State-Dependent Fading Gaussian Channel with Common Reconstruction Constraints}
\author{Viswanathan Ramachandran
\thanks{The author is with the Department of Electrical Engineering, Indian Institute of Technology (IIT), Jodhpur, India (email: vramachandran@iitj.ac.in).}
}
\maketitle

\begin{abstract}
The task of jointly communicating a message and reconstructing a common estimate of the channel state is examined for a fading Gaussian model with additive state interference. The state is an independent and identically distributed (i.i.d.) Gaussian sequence known non-causally at the transmitter, and the instantaneous fading coefficient is perfectly known at both the transmitter and the receiver. The receiver is required to decode the transmitted message and, in addition, reconstruct the state under a common reconstruction (CR) constraint ensuring that its estimate coincides with that at the transmitter. A complete characterization of the optimal rate–distortion tradeoff region for this setting is the main result of our work. The analytical results are also validated through numerical examples illustrating the rate–distortion and power-distortion tradeoffs.
\end{abstract}


\IEEEpeerreviewmaketitle

\section{Introduction}
State-dependent channels serve as a powerful model for communication scenarios where the channel behavior is influenced by an underlying random process that is only partially known. The achievable rate in such systems depends critically on how much of the state information is available at the transmitter and the receiver. A well-known instance is the case where the transmitter has non-causal state access while the receiver does not, originally analyzed by \cite{gel1980coding} for discrete memoryless channels and by \cite{costa1983writing} for the Gaussian counterpart.

Reference~\cite{sutivong2005channel} extended \cite{costa1983writing} to a scenario where the transmitter not only sends a message but also helps the receiver estimate the channel state. However, the receiver’s state estimate in that model is unknown to the transmitter. Building on this idea,~\cite{steinberg2009coding} introduced a common reconstruction (CR) constraint, ensuring identical reconstructions at both terminals. This is an important feature in applications such as medical or biometric data transmission (providing both entities with a shared reference for synchronization).

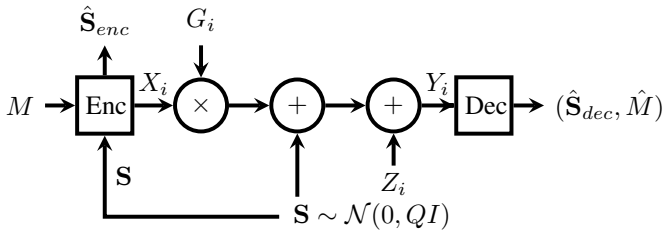
\begin{figure}[h]
\begin{center}
\begin{tikzpicture}[>=stealth,scale=0.8, line width=1.5pt]
\draw (0,0) node[rectangle, draw, text width=0.5cm,minimum height=0.75cm] (enc){Enc};
\draw (3.2,0) node[circle, draw] (state) {$+$};
\draw (1.6,0) node[circle, draw] (fading) {$\times$};
\draw (4.8,0) node[circle, draw] (n1) {$+$};
\draw (6.3,0) node[rectangle, draw, text width=0.5cm,minimum height=0.75cm] (dec1){Dec};
\draw[<-] (enc.west) --++(-0.5,0) node[left]{$M$};
\draw[->] (enc) --node[above]{$X_i$} (fading);
\draw[->] (fading) -- (state);
\draw[->] (state) --  (n1);
\draw[->] (n1) --node[above]{$Y_i$}++(1,0) --  (dec1);
\draw[->] (dec1) --++(1,0) node[right]{$(\bhS_{dec},\hat{M})$};
\draw (3.2,-1.85) node (st) {\phantom{S}};
\draw (st) ++(-0.25,0) node[right, rectangle, fill=white] {$\bS \sim \mathcal{N}(0,Q I)$}; 
\draw [->] (st) -- (state);
\draw [->] (st) -| node[right,pos=0.75]{$\bS$}(enc);
\draw[<-] (n1) --++(0,-1) node[below,yshift=0.05cm]{$Z_i$};
\draw[<-] (fading) --++(0,1) node[above,yshift=0.05cm]{$G_i$};
\draw[->] (enc) --++(0,1) node[above]{$\bhS_{enc}$};
\end{tikzpicture}
\end{center}
\label{fig:model}
\caption{State-dependent fading channel with Common Reconstructions\label{fig:bc2}}
\end{figure}

The studies in~\cite{steinberg2009coding} have inspired numerous extensions in recent years. The work of \cite{ramachandran2018state} investigates a state-dependent Gaussian broadcast channel (BC) where both receivers must form a CR of the additive state. On the source coding front,~\cite{ahmadi2013heegard} formulated the classical Heegard–Berger (HB) and cascade source coding problems under a CR constraint. Related extensions include the treatment of broadcasting with CR in relay networks~\cite{timo2010lossy}, and successive refinement under CR~\cite{vellambi2014successive}. The vector Gaussian counterpart of the HB setting was analyzed in~\cite{lu2021vector}, which established the rate–distortion region under CR. More recently,~\cite{adikari2022two} considered two-terminal source coding where each terminal observes one of two correlated sources and seeks to reconstruct their sum under a CR constraint. This was later generalized in~\cite{adikari2024common} to allow reconstruction of arbitrary common functions of the sources.

However, fading is a fundamental impairment in wireless channels that must be explicitly modeled. This paper addresses that aspect by studying joint message transmission and common reconstruction over a fading state-dependent channel. Prior work on such channels has largely ignored state estimation/ CR -- for instance, \cite{vaze2008dirty} characterized high-SNR rates with partial fading knowledge. A notable exception is the work of \cite{ramachandran2021joint}, which analyzed joint communication and state estimation for fading Gaussian channels with non-causal state information at the transmitter, but without CR. The present work extends the state estimation work of \cite{ramachandran2021joint} to CR constraints, and the framework of~\cite{steinberg2009coding} to fading scenarios. The principal contribution of this work is the complete characterization of the optimal trade-off between the message transmission rate and the state estimation distortion (Theorem~\ref{thm:mainN} in Sect.~\ref{sec:sys}). 

\emph{Notations:} Random variables are denoted by upper-case letters, with their realizations in lower case letters. 
A sequence $(A_1,A_2,\cdots,A_n)$ is denoted by $A^n$ or equivalently as the boldface vector $\mathbf{A}$. $||\cdot||$ denotes the vector Euclidean norm.

\section{Channel Model and Key Results} \label{sec:sys}
The system illustrated in Fig.~1 depicts a single-user state-dependent fading Gaussian channel. The transmitter has non-causal access to the state sequence $S^n$, while the instantaneous fading coefficient $G_i$, $i \in [1,2,\ldots,n]$, is perfectly known at both ends.
A block-fading model is considered, with the input constrained in average power across fading blocks. The central aim is to enable \emph{common reconstruction} (CR) of the state process, meaning that the receiver’s reconstruction of the state must exactly match the transmitter’s version. Alongside this requirement, the transmitter also communicates an independent message $M$ to the receiver. The problem is to characterize the optimal trade-off between the achievable rate $R$ and the mean distortion $D$ in reconstructing the state over the fading process. Each channel use is described by
\begin{align}
Y_i = G_i X_i + S_i + Z_i, \quad i \in [1,2,\ldots,n], \label{eq:model}
\end{align}
where $G^n$, $S^n$, and $Z^n$ denote the fading, state, and noise processes, respectively, and are mutually independent. The state sequence $S^n$ is available only at the encoder, whereas the fading gains $G_i$ are known to both terminals. Both $S_i$ and $Z_i$ are independent and identically distributed (i.i.d.) zero-mean Gaussian variables with variances $Q$ and $\sigma_z^2$, respectively.  

The message $M$ is uniformly distributed over the set $\{1,2,\ldots,2^{nR}\}$. The transmitter is subject to an average power constraint across channel uses and fading realizations:
\begin{equation*}
    \frac{1}{n} \, \mathbb{E}_{G}\!\left[\mathbb{E}_{S^n}\!\left[\sum_{i=1}^{n} X_i^2(m,G,S^n)\right]\right] \leq \mathrm{P}, 
    \:\: \forall \, m \in \{1,2,\ldots,2^{nR}\}.
\end{equation*}

\begin{defn}
An $(n,R,D,\epsilon)$ common reconstruction (CR) scheme is defined by an encoder map $\mathcal{E}:\{1,2,\cdots,2^{nR}\} \times \bhcS \times \mathcal{G} \to \mathcal{X}$, a sender quantization map $\phi:\bhcS \times \mathcal{G} \times \{1,2,\cdots,2^{nR}\} \to \mathbb R$, a decoding map $\psi:\bhcY \times \mathcal{G}^n \to \{1,2,\cdots, 2^{nR}\}$, and a receiver reconstruction map $\zeta: \bhcY \times \mathcal{G}^n \to \mathbb R^n$. For a message $M$ uniformly distributed over $\{1,2,\cdots,2^{nR}\}$, let $\bX=\{\mathcal{E}(M,\bS, G_i)\}_{i=1}^n$, $\bhS_{enc}=\{\phi(M,\bS,G_i)\}_{i=1}^n$, and $\bhS_{dec}=\zeta(\bY,G^n)$. The following conditions must then hold:
\begin{equation}
\eE{[||\bS-\zeta(\bY,G^n)||^2]} \leq n(D+\epsilon),
\end{equation}
\begin{equation}
\mathbb{P}(\psi(\bY,G^n) \neq M) \leq \epsilon,
\end{equation}
\begin{equation}
\mathbb{P}\left(\zeta(\bY,G^n) \neq \bhS_{enc} \right) \leq \epsilon,
\end{equation}
subject to the average power constraint $\eE ||\bX||^{2} \leq n\mathrm{P}$. 
\end{defn}

We say that a pair $(R,D)$ is \emph{achievable} if, for every $\epsilon > 0$, there exists an $(n,R,D,\epsilon)$ common reconstruction scheme, possibly for sufficiently large $n$. The capacity region $\sCapa$ is defined as the convex closure of all achievable $(R,D)$ pairs, with $0 \leq D \leq Q$. For notational convenience, we parameterize $\sCapa$ in terms of three parameters $(\rho_1,\rho_2,d)$, where $\bar{\rho} = (\rho_1, \rho_2)$ satisfies $\rho_1^2 + \rho_2^2 \leq 1$, and $d$ satisfies $0 \leq d \leq Q$. The quantity $R(\bar{\rho}, d)$ is defined as in~\eqref{eq:ra} at the beginning of the next page. Let $\kappa$ denote the set of all $(\rho_1, \rho_2, d)$ such that $R(\bar{\rho}, d) \geq 0$. 
The main result of this paper is presented next.

\begin{theorem} \label{thm:mainN}
The capacity region $\sCapa$ is completely characterized by the convex hull in $\mathbb R_+^2$ of 
\begin{align*}
\underset{D \in \kappa}{\bigcup} \left\{ \left( R, D \right)\Bigg\vert \begin{aligned} &\exists \:\: (\rho_1,\rho_2,{d}) \in \kappa \:\: \textup{s.t.}\:\: R \leq \eE_{G}[ R(\rho_1,\rho_2, d)],\\
& D \geq d \end{aligned}   \right\},
\end{align*}
where the expectation is over the fading distribution, along with the power constraint $\mathbb{E}_{G}[P(G)] \leq \mathrm{P}$.
\end{theorem}

\begin{figure*}
\hrule
\begin{gather}
R(\bar{\rho},d) = \frac{1}{2}\log\Bigg(\frac{d(G^2 P(G)+Q+\sigma_z^2+ 2 G\rho_1\sqrt{P(G)(Q-d)}+2G \rho_2\sqrt{P(G) d})}{Q((1-\rho_1^2) G^2 P(G)
+d+\sigma_z^2+2G\rho_2\sqrt{P(G)d})
}\Bigg). \label{eq:ra}
\end{gather}
\hrule
\end{figure*}
\begin{proof}
The achievability proof is presented in Section~\ref{sec:achieve}, while the converse is established in Section~\ref{sec:conv}.
\end{proof}
%


The following mild assumption is used in our converse:
\begin{equation}
H(\bhS_{enc}|\bhS_{dec}) \leq n \epsilon_n, \tag{T1} \label{eq:tec}
\end{equation}
where $\epsilon_n \to 0$ as $n \to \infty$. This entails no loss of generality and allows setting $\bhS_{enc} = \bhS_{dec} = \bhS$ in the proof. Intuitively, \eqref{eq:tec} restricts attention to an exponential number of agreed reconstructions. In particular, if a scheme without reproduction cardinality bounds achieves $(R,D)$, then for any $\delta > 0$, the pair $(R,D+\delta)$ can be achieved using reconstructions from an alphabet of size $2^{nc(\delta)}$. This follows, for example, by quantizing the reconstructions using a scalar quantizer with distortion $\delta$ and rate scaling as $2^{nc(\delta)}$.

\section{Achievability proof of Theorem~\ref{thm:mainN}} \label{sec:achieve}
For a given fading instantiation $G = g$, a channel use is
\begin{align}
Y = g X + S + Z, \label{eq:modelsplit}
\end{align}
with $\mathbb{E}[X^2] \leq P(G)$ and $\mathbb{E}_{G}[P(G)] \leq \mathrm{P}$.
The achievability of $\sCapa$ follows from standard Gaussian random coding arguments using the Gel’fand–Pinsker (GP) framework~\cite{gel1980coding}. An auxiliary variable $U$ is introduced to jointly handle message transmission and state reconstruction: the sequence $U^n$ decoded at the receiver simultaneously carries the information message and enables common reconstruction of the state.

\begin{lemma} \label{lem:achieve}
For any distribution $p(u,x|s)$ satisfying ${I(U;Y) \geq I(U;S)}$ and $\mathbb{E}[X^2] \leq P(G)$, all rate–distortion pairs $(R,D)$ meeting the following constraints are achievable:
\begin{align}
&0 \leq  R \leq I(U;Y) - I(U;S), \label{eq:achieve:rate} \\
&D  \geq \eE \left( S - \eE \left[ S|U \right] \right)^2. \label{eq:achieve:dist}
\end{align}
\end{lemma}

\begin{proof}
The discrete memoryless version of this coding scheme appears in~\cite[Theorem~2]{steinberg2009coding}. Lemma~\ref{lem:achieve} adapts that construction to the Gaussian case with squared-error distortion. Under the rate constraints of Lemma~\ref{lem:achieve}, $u^n$ can be decoded with vanishing probability of error, and the desired distortion is achieved by forming the decoder’s estimate on a per-letter basis.  
\end{proof}

Let $(\rho_1, \rho_2, d)$ be such that $(\rho_1, \rho_2, d) \in \kappa$, i.e., $\rho_1^2 + \rho_2^2 \leq 1$, $0 \leq d \leq Q$, and the rate function $R(\bar{\rho}, d)$ in~\eqref{eq:ra} is non-negative. Let us choose the conditional distribution $p(u,x|s)$ employed in Lemma~\ref{lem:achieve}.
We express the state variable as
\begin{equation}
S = U + T, \label{eqs}
\end{equation}
where $U$ and $T$ are Gaussian random variables (independent) with zero mean and variances $\Var(U) = Q - d$ and $\Var(T) = d$. This construction determines $p(u|s)$. The conditional distribution $p(x|u,s)$ is then defined by
\begin{align} \label{eq:inp:1}
X = \rho_1 \sqrt{\frac{P(g)}{\Var(U)}} U
+ \rho_2 \sqrt{\frac{P(g)}{\Var(T)}} T.
\end{align}
Since $\rho_1^2 + \rho_2^2 \leq 1$, it follows that $\mathbb{E}[X^2] \leq P(g)$. For reference in the outer bound analysis, we note that the covariance matrix of the jointly Gaussian $(X, U, T)$ defined above is
\begin{equation}
\left(
\begin{array}{ccc}
K_{00} & K_{01} & K_{02}\\
K_{01} & K_{11} & 0\\
K_{02} & 0 & K_{22}
\end{array}
\right),
\end{equation}
where $K_{00} \leq P(G)$, $K_{11} = Q - d$, $K_{22} = d$, and $\rho_k = K_{0k}/\sqrt{K_{00}K_{kk}}$ for $k = 1, 2$. Now consider
\begin{align}
    R &= I(U;Y) - I(U;S) \notag\\
      &= h(Y) - h(Y|U) - h(S) + h(S|U). \label{eq:achieve:rb:1}
\end{align}
Since $(X,U,S)$ are jointly Gaussian, and defining the function $f(x) = \tfrac{1}{2}\log(2\pi e x)$, we can write
\begin{align}  \label{eq:hat:ps:2}
    &h(Y|U) = f(\sigma_{Y|U}^2) \triangleq f\bigl(\min_{\alpha} \mathbb{E}\!\left[(Y - \alpha U)^2\right]\bigr), \\
    &\phantom{ww}= (1 - \rho_1^2)G^2 P(G)
       + d + \sigma_z^2
       + 2G\rho_2\sqrt{P(G)d}.
\end{align}
Similarly, $h(Y)$ and $h(S|U)$ can be computed, and substituting these expressions into~\eqref{eq:achieve:rb:1} yields~\eqref{eq:ra}. 
Since $(\rho_1, \rho_2, d) \in \kappa$ ensures that $R \geq 0$, the chosen $p(u,x|s)$ satisfies $I(U;Y) \geq I(U;S)$ as required in Lemma~\ref{lem:achieve}. 
Furthermore, from the choice of $p(u|s)$, we have $\mathbb{E}\!\left[\left(S - \mathbb{E}[S|U]\right)^2\right] = d.$ 

\section{Converse proof of Theorem~\ref{thm:mainN}} \label{sec:conv}
\def\va{\boldsymbol{\alpha}}%

\def\neps{n \epsilon_n}
If a rate–distortion pair $(R,D)$ is achievable, then for any $\epsilon > 0$ there exists an $(n,R,D,\epsilon)$ CR scheme.
For such a scheme, define $K_i$ as the covariance matrix of the vector $(X_i,\hat{S}_{i},S_i-\hat{S}_{i})$ under the fading instantiation $G_i=G$, for each $i=1,2,\ldots,n$. The average covariance matrix is then defined as $K=\frac{1}{n}\sum_{i=1}^{n}K_i$. Without loss of generality, we restrict attention to zero-mean random variables that satisfy standard orthogonality relations below; any deviation from these would only increase the transmit power or distortions without improving the achievable rate or reliability:
\begin{align*}
\eE [ (S_i-\hat{S}_{i})\hat{S}_{i} ] &= 0, \:\: i=1,2,\ldots,n.
\end{align*}
From these definitions, the covariance matrix is:
\begin{equation}
    K =
    \left(
    \begin{array}{ccc}
        K_{00} & K_{01} & K_{02}\\
        K_{01} & K_{11} & 0\\
        K_{02} & 0 & K_{22}
    \end{array}
    \right), \label{eq:Kconverse}
\end{equation}
where the following conditions arise from the orthogonality relationships and the average power constraint:
\begin{gather}
    K_{00} \leq P(G), \quad K_{11} + K_{22} = Q, \label{eq:ConverseConditions1}\\
    K_{22} \leq D + \epsilon. \label{eq:ConverseConditions2}
\end{gather}
Define $\rho_k = K_{0k}/\sqrt{K_{00}K_{kk}}$ for $k = 1,2$, interpreting $\tfrac{0}{0}$ as~0. 
Since $K$ is positive semi-definite, we have ${\bf a}^T K {\bf a} \geq 0$ for any vector~${\bf a}$. 
Choosing ${\bf a}^T = (-1, \tfrac{K_{01}}{K_{11}}, \tfrac{K_{02}}{K_{22}})$, we obtain:
\begin{align}
    \rho_1^2 + \rho_2^2 \leq 1. \label{eq:ConverseConditions5}
\end{align}
Let $\bhS$, which depends on $(\bS, M, G^n)$, represent the encoder’s quantized reconstruction of the state sequence. Invoking the assumption in~\eqref{eq:tec} and applying Fano’s inequality, we obtain
\begin{align}
    H(\bhS,M \mid \bY) \leq n \epsilon_n, \label{eq:fano:1}
\end{align}
where $\epsilon_n \to 0$ as $n \to \infty$. 
We denote the pair $(\bhS, M)$ by $\btS$, and proceed to establish the following chain of inequalities.
\begin{align}
&nR \stackrel{(a)}=   H(M) \stackrel{(b)}= H(M|G^n) \notag\\
    &\phantom{w}= H(\bhS,M|G^n) - H(\bhS|M,G^n) \notag\\
    &\phantom{w}\stackrel{(c)}= H(\bhS,M|G^n) - H(\bhS|M,G^n)+H(\bhS|M,\bS,G^n) \notag\\
     &\phantom{w}= H(\bhS,M|G^n) - I(\bhS;\bS|M,G^n) \notag\\
     &\phantom{w}\stackrel{(d)}= H(\bhS,M|G^n) \!-\! I(\bhS,M;\bS|G^n) = H(\btS|G^n) \!-\! I(\btS;\bS|G^n) \notag\\
     &\phantom{w}= H(\btS|G^n)-H(\btS|\bY,G^n)+H(\btS|\bY,G^n) - I(\btS;\bS|G^n) \notag\\
     &\phantom{w}\stackrel{(e)}\leq I(\btS;\bY|G^n) + H(\btS|\bY) - I(\btS;\bS|G^n) \notag\\
     &\phantom{w}\stackrel{(f)}\leq I(\btS;\bY|G^n) - I(\btS;\bS|G^n) + n\epsilon_n \notag\\
&\phantom{w}= h(\bY|G^n) - h(\bY|\btS,G^n)- h(\bS)  +  h(\bS|\btS,G^n)+ n\epsilon_n \notag\\
&\phantom{w}= h(\bY|G^n) - h(\bY|\btS,\bS,G^n)- h(\bS)+  h(\bS|\btS,\bY,G^n)+ n\epsilon_n \notag\\
&\phantom{w}\stackrel{(g)}\leq h(\bY|G^n) - h(\bZ)- h(\bS)+  h(\bS|\bhS,\bY,G^n)+ n\epsilon_n \notag\\
&\phantom{w}\stackrel{(h)}\leq \sum_{i=1}^n \Bigl(h(Y_i|G_i)-h(Z_i)-h(S_i)+h(S_i|\hat{S}_i,Y_i,G_i)\Bigr) + n\epsilon_n \notag\\
&\phantom{w}= \eE_{G}\Bigl[\sum_{i=1}^n \Bigl(h(Y_i|G_i=G)-h(Z_i)-h(S_i) \notag\\
&\phantom{wwwwwwww}+h(S_i|\hat{S}_i,Y_i,G_i=G)\Bigr)\Bigr] + n\epsilon_n,
	\label{eq:r1mur2:1}
\end{align}

\begin{figure*}
\hrule
\begin{align}
R &\leq \frac{1}{2}\log\Bigg(\frac{K_{22}(G^2 K_{00} +K_{11}+K_{22}+\sigma_z^2+  2G \rho_1\sqrt{K_{00}K_{11}} +  2G \rho_2\sqrt{K_{00}K_{22}})}{(K_{11}+K_{22})(G^2 (1-\rho_1^2) K_{00} +K_{22}+\sigma_z^2 +  2G \rho_2\sqrt{K_{00}K_{22}})
}\Bigg).
\label{eq:ConverseRupperbound}
\end{align}
\hrule
\end{figure*}

\noindent where (a) follows since $M$ is uniformly distributed on $\{1,2,\ldots,2^{nR}\}$, (b) follows from the independence of $M$ and $G^n$, (c) follows since $\bhS$ is determined by $(M,\bS,G^n)$, (d) follows from the independence of $M$ and $(\bS,G^n)$, (e) follows since conditioning does not increase the entropy, (f) follows by Fano's inequality~\eqref{eq:fano:1}, (g) follows since given $(\bS,M,G^n)$, the residual uncertainty in $\bY$ is only that of the noise $\bZ$, and (h) follows from the independence bound on entropy.

To further upper bound~\eqref{eq:r1mur2:1}, we define the function $f(x) = \tfrac{1}{2}\log(2\pi e x)$. 
Then, for any real numbers $(\alpha_1, \alpha_2)$, the maximum entropy property of Gaussian random variables yields
\begin{align*}
    &\sum_{i=1}^n h(S_i \mid \hat{S}_i, Y_i, G_i = G) 
    \leq \sum_{i=1}^n h(S_i - \alpha_1 \hat{S}_i - \alpha_2 Y_i) \notag\\
    &\leq \sum_{i=1}^n 
    f\Bigl(
        \mathbb{E}[S_i^2] 
        + \alpha_1^2 \mathbb{E}[\hat{S}_i^2]
        + \alpha_2^2 \mathbb{E}[Y_i^2] \notag\\
    &\phantom{\leq \sum_{i=1}^n f(}
        - 2\alpha_1 \mathbb{E}[S_i \hat{S}_i]
        - 2\alpha_2 \mathbb{E}[S_i Y_i]
        + 2\alpha_1 \alpha_2 \mathbb{E}[\hat{S}_i Y_i]
    \Bigr) \notag\\
    &\leq n\, f\Bigl(
        \frac{1}{n} \sum_{i=1}^n 
        \Bigl\{
            \mathbb{E}[S_i^2]
            + \alpha_1^2 \mathbb{E}[\hat{S}_i^2]
            + \alpha_2^2 \mathbb{E}[Y_i^2] \notag\\
    &\phantom{\leq n\, f(}
            - 2\alpha_1 \mathbb{E}[S_i \hat{S}_i]
            - 2\alpha_2 \mathbb{E}[S_i Y_i]
            + 2\alpha_1 \alpha_2 \mathbb{E}[\hat{S}_i Y_i]
        \Bigr\}
    \Bigr).
\end{align*}

\noindent with the ultimate step being justified by Jensen's inequality. Choosing $(\alpha_1, \alpha_2)$ to minimize the above and using~\eqref{eq:Kconverse}:
\begin{align}
    &\sum_{i=1}^n h(S_i \mid \hat{S}_i, Y_i, G_i = G) 
    \leq n\, f(\sigma_{S|\hat{S},Y}^2), 
    \quad \text{where} \label{eq:r1mur2:2}\\
    &\sigma_{S|\hat{S},Y}^2 
    = \frac{K_{22} \sigma_z^2}{
        G^2 (1 \!-\! \rho_1^2) K_{00} 
        \!+\! K_{22} 
        \!+\! \sigma_z^2 
        \!+\! 2G \rho_2 \sqrt{K_{00} K_{22}}
    }. \notag
\end{align}

Similarly, the other term in~\eqref{eq:r1mur2:1} can be upper bounded as 
$\sum_{i=1}^n h(Y_i|G_i=G) \leq n\, f(\sigma_Y^2)$, 
where
\begin{align}
    \sigma_Y^2 
    &= G^2  K_{00} + K_{11} + K_{22} + \sigma_z^2 
    + 2G  \rho_1  \sqrt{K_{00}  K_{11}} \notag\\
    &\phantom{wwww}+ 2 G  \rho_2  \sqrt{\!K_{00} \! K_{22}}. \label{eq:outb1} 
\end{align}

Using~\eqref{eq:r1mur2:1}--\eqref{eq:outb1} and letting $n \to \infty$ so that $\epsilon_n \to 0$, we conclude that $(R, D)$ is achievable only if there exist 
$K_{jj} \geq 0$ for $j = 0, 1, 2$ and $\rho_1, \rho_2 \in [-1,1]$ obeying~\eqref{eq:ConverseConditions1}--\eqref{eq:ConverseConditions5} and~\eqref{eq:ConverseRupperbound} (shown at the beginning of the following page). 
Combining~\eqref{eq:ConverseConditions1}--\eqref{eq:ConverseConditions5} 
with~\eqref{eq:ConverseRupperbound}, the converse proof of Theorem~\ref{thm:mainN} is complete, as the outer bound takes the same functional form as expression~\eqref{eq:ra} in the achievability proof.


\noindent \textbf{Numerical Example:} The fading process $G^n$ is modeled as an i.i.d.\ Rayleigh sequence with probability density
\begin{align*}
p_{G}(g) = 2g e^{-g^2}, \quad g \ge 0.
\end{align*}
Using this model, Fig.~2 depicts the rate–distortion region (blue curve) for parameter values $\mathrm{P}=2.5$, $Q=1$, and $\sigma_Z^2=1$.
For comparison, Fig.~\ref{fig:rd} (red curve) shows the rate–distortion tradeoff for a non-fading channel with the same average signal-to-noise ratio (SNR). 
Theorem~\ref{thm:mainN} (shown via the blue curve) exhibits a compressed rate–distortion region relative to the static (non-fading) channel, showing that fading reduces achievable rates uniformly across all distortion levels. The loss in rate relative to the non-fading case arises primarily because the transmitter cannot fully exploit instantaneous channel variations when enforcing common reconstruction, even with perfect knowledge of the fading coefficients. 
This highlights that fading introduces an inherent penalty on both communication and estimation efficiency due to ergodic averaging under the common reconstruction constraint.
\begin{figure}[]
\begin{center}
\vspace{-6mm}
\scalebox{1.0}{\begin{tikzpicture}
  \begin{axis}[
    width=1.05\columnwidth,
height=0.8\columnwidth,
    xmin=0, xmax=1.0,
    ymin=0.1, ymax=0.8,
    xlabel={Rate $R$ (bits/channel-use)},
    ylabel={Distortion $D$},
    xmajorgrids=true, ymajorgrids=true,
    grid style={dashed, gray!30},
    legend pos=south east,
    label style={font=\small},
    ticklabel style={font=\small},
    legend style={
legend pos=north west,legend cell align=left,row sep=-0.5ex,grid style={dashed},font=\small
}
  ]

  \addplot [mark=triangle*, mark size=1.5,  mark repeat=5, mark phase=4, line width=1.0pt, blue] table [x=R, y=D] {%
R D
-1.49376881 0.02000000
-1.25307432 0.02806700
-1.07499918 0.03613400
-0.93305988 0.04420200
-0.81572308 0.05226900
-0.71601957 0.06033600
-0.62952133 0.06840300
-0.55295184 0.07647100
-0.48411627 0.08453800
-0.42203501 0.09260500
-0.36558358 0.10067200
-0.31389719 0.10873900
-0.26629638 0.11680700
-0.22204185 0.12487400
-0.18047096 0.13294100
-0.14165914 0.14100800
-0.10530157 0.14907600
-0.07114140 0.15714300
-0.03896021 0.16521000
-0.00857067 0.17327700
0.02018909 0.18134500
0.04770515 0.18941200
0.07408564 0.19747900
0.09919586 0.20554600
0.12313242 0.21361300
0.14598125 0.22168100
0.16781910 0.22974800
0.18871483 0.23781500
0.20873044 0.24588200
0.22792197 0.25395000
0.24634021 0.26201700
0.26433498 0.27008400
0.28171535 0.27815100
0.29844094 0.28621800
0.31454730 0.29428600
0.33006710 0.30235300
0.34503042 0.31042000
0.35946503 0.31848700
0.37339663 0.32655500
0.38684902 0.33462200
0.39984431 0.34268900
0.41240304 0.35075600
0.42454436 0.35882400
0.43655271 0.36689100
0.44820773 0.37495800
0.45948976 0.38302500
0.47041350 0.39109200
0.48099273 0.39916000
0.49124039 0.40722700
0.50116863 0.41529400
0.51078888 0.42336100
0.52011190 0.43142900
0.52914783 0.43949600
0.53790621 0.44756300
0.54639608 0.45563000
0.55462592 0.46369700
0.56260379 0.47176500
0.57039333 0.47983200
0.57812266 0.48789900
0.58561740 0.49596600
0.59288400 0.50403400
0.59992857 0.51210100
0.60675690 0.52016800
0.61337445 0.52823500
0.61978639 0.53630300
0.62599760 0.54437000
0.63201269 0.55243700
0.63783603 0.56050400
0.64347171 0.56857100
0.64892361 0.57663900
0.65419536 0.58470600
0.65929038 0.59277300
0.66421188 0.60084000
0.66896283 0.60890800
0.67354603 0.61697500
0.67796407 0.62504200
0.68233794 0.63310900
0.68660139 0.64117600
0.69070293 0.64924400
0.69464437 0.65731100
0.69842732 0.66537800
0.70205322 0.67344500
0.70552327 0.68151300
0.70883850 0.68958000
0.71199971 0.69764700
0.71500750 0.70571400
0.71786223 0.71378200
0.72056405 0.72184900
0.72311285 0.72991600
0.72550826 0.73798300
0.72774963 0.74605000
0.72983603 0.75411800
0.73176620 0.76218500
0.73353853 0.77025200
0.73515103 0.77831900
0.73660131 0.78638700
0.73788648 0.79445400
0.73900316 0.80252100
0.73994738 0.81058800
0.74071450 0.81865500
0.74129911 0.82672300
0.74169495 0.83479000
0.74189472 0.84285700
0.74193153 0.85092400
0.74179500 0.85899200
0.74142805 0.86705900
0.74081683 0.87512600
0.73994497 0.88319300
0.73879290 0.89126100
0.73733707 0.89932800
0.73554881 0.90739500
0.73339271 0.91546200
0.73082432 0.92352900
0.72778674 0.93159700
0.72420527 0.93966400
0.71997870 0.94773100
0.71496431 0.95579800
0.70894958 0.96386600
0.70159304 0.97193300
0.69227897 0.98000000
};
  \addlegendentry{Theorem~\ref{thm:mainN}}

  \addplot [mark=triangle*, mark size=1.5,  mark repeat=5, mark phase=4, line width=1.0pt, red] table [x=R, y=D] {%
R D
-1.36430850 0.02000000
-1.12331380 0.02806700
-0.94372842 0.03613400
-0.80128186 0.04420200
-0.68364061 0.05226900
-0.58329345 0.06033600
-0.49551396 0.06840300
-0.41809474 0.07647100
-0.34896563 0.08453800
-0.28662212 0.09260500
-0.22939580 0.10067200
-0.17668835 0.10873900
-0.12812494 0.11680700
-0.08315704 0.12487400
-0.04133867 0.13294100
-0.00230230 0.14100800
0.03473299 0.14907600
0.06973268 0.15714300
0.10271310 0.16521000
0.13386388 0.17327700
0.16334896 0.18134500
0.19131102 0.18941200
0.21787504 0.19747900
0.24341705 0.20554600
0.26804832 0.21361300
0.29156132 0.22168100
0.31403404 0.22974800
0.33553650 0.23781500
0.35613178 0.24588200
0.37587689 0.25395000
0.39482356 0.26201700
0.41301885 0.27008400
0.43091416 0.27815100
0.44818323 0.28621800
0.46480905 0.29428600
0.48082500 0.30235300
0.49626183 0.31042000
0.51114796 0.31848700
0.52550966 0.32655500
0.53937131 0.33462200
0.55275555 0.34268900
0.56568343 0.35075600
0.57848176 0.35882400
0.59093339 0.36689100
0.60297612 0.37495800
0.61462608 0.38302500
0.62589838 0.39109200
0.63680720 0.39916000
0.64736585 0.40722700
0.65758684 0.41529400
0.66748194 0.42336100
0.67706224 0.43142900
0.68633817 0.43949600
0.69531960 0.44756300
0.70433072 0.45563000
0.71306409 0.46369700
0.72152336 0.47176500
0.72971642 0.47983200
0.73765073 0.48789900
0.74533334 0.49596600
0.75277093 0.50403400
0.75996979 0.51210100
0.76693591 0.52016800
0.77367491 0.52823500
0.78019213 0.53630300
0.78649259 0.54437000
0.79258105 0.55243700
0.79846804 0.56050400
0.80440381 0.56857100
0.81013526 0.57663900
0.81566613 0.58470600
0.82099992 0.59277300
0.82613991 0.60084000
0.83108916 0.60890800
0.83585049 0.61697500
0.84042652 0.62504200
0.84481964 0.63310900
0.84903205 0.64117600
0.85306571 0.64924400
0.85692237 0.65731100
0.86060358 0.66537800
0.86411067 0.67344500
0.86744473 0.68151300
0.87060663 0.68958000
0.87359702 0.69764700
0.87658160 0.70571400
0.87941846 0.71378200
0.88207960 0.72184900
0.88456460 0.72991600
0.88687274 0.73798300
0.88900294 0.74605000
0.89095382 0.75411800
0.89272358 0.76218500
0.89431004 0.77025200
0.89571055 0.77831900
0.89692196 0.78638700
0.89794057 0.79445400
0.89876205 0.80252100
0.89938137 0.81058800
0.89979271 0.81865500
0.89998933 0.82672300
0.89996341 0.83479000
0.89970593 0.84285700
0.89920641 0.85092400
0.89845266 0.85899200
0.89743045 0.86705900
0.89612304 0.87512600
0.89451066 0.88319300
0.89256974 0.89126100
0.89027193 0.89932800
0.88758273 0.90739500
0.88445956 0.91546200
0.88087406 0.92352900
0.87674929 0.93159700
0.87196669 0.93966400
0.86640279 0.94773100
0.85988349 0.95579800
0.85214961 0.96386600
0.84278470 0.97193300
0.83109493 0.98000000
};
\addlegendentry{No fading}
\end{axis}
\end{tikzpicture}}
\end{center}
\caption{Illustration of the rate-distortion region in Theorem \ref{thm:mainN}. \label{fig:rd}}
\end{figure}
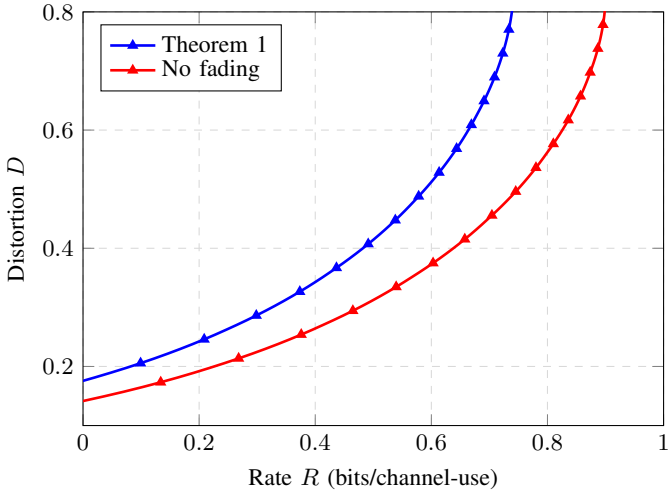

To characterize the distortion region, we can also equivalently compute the minimum average transmit power required to achieve a given distortion $D$ at the receiver while supporting a message rate $R$, denoted by $\mathrm{P}(R,D)$. The resulting power–distortion tradeoff for different values of $R$ is illustrated in Fig.~\ref{fig:rd1}, for the case $Q=1$ and $\sigma_z^2=1$. From the plots, it is observed that for a fixed distortion, achieving a higher rate requires higher transmit power, while for a fixed rate, attaining lower distortion also demands increased power. The convexity of the tradeoff illustrates diminishing returns, in that achieving very low distortion demands a disproportionately large increase in transmit power. The monotonic trend across $R$ further confirms the fundamental coupling between information transmission and estimation fidelity inherent to the CR constraint. Together, these results demonstrate how fading impacts the communication–reconstruction tradeoff for CR.
\begin{figure}[]
\begin{center}
\vspace{-6mm}
\scalebox{1.0}{\begin{tikzpicture}
  \begin{axis}[
    width=1.05\columnwidth,
height=0.8\columnwidth,
    xmin=0, xmax=0.95, ymin=1e-4, ymax=100,
    ymode=log,
    xlabel={Distortion $D$}, ylabel={Required power $\mathrm{P}$},
    xmajorgrids=true, ymajorgrids=true,
    grid style={dashed, gray!30}, legend pos=south west,
    label style={font=\small}, ticklabel style={font=\small}, legend style={font=\small}]

  \addplot [thick, mark=*, color=orange] table [x=D, y=P] {%
D P
0.020000 69.976837
0.044615 22.815019
0.069231 11.785977
0.093846 7.415864
0.118462 5.017185
0.143077 3.637792
0.167692 2.722265
0.192308 2.111914
0.216923 1.672462
0.241538 1.348976
0.266154 1.110939
0.290769 0.921730
0.315385 0.775246
0.340000 0.659279
0.364615 0.561623
0.389231 0.482277
0.413846 0.421242
0.438462 0.366311
0.463077 0.323586
0.487692 0.286965
0.512308 0.250344
0.536923 0.219826
0.561538 0.195412
0.586154 0.177102
0.610769 0.152688
0.635385 0.140481
0.660000 0.122170
0.684615 0.109963
0.709231 0.097756
0.733846 0.085549
0.758462 0.073342
0.783077 0.061135
0.807692 0.055032
0.832308 0.048928
0.856923 0.036721
0.881538 0.030618
0.906154 0.024514
0.930769 0.018411
0.955385 0.012307
0.980000 0.006204
};
  \addlegendentry{$R=0.00$}

  \addplot [thick, mark=*, color=red] table [x=D, y=P] {%
D P
0.044615 34.649724
0.069231 17.999350
0.093846 11.541837
0.118462 8.154389
0.143077 6.012057
0.167692 4.663181
0.192308 3.717137
0.216923 3.027441
0.241538 2.526953
0.266154 2.124121
0.290769 1.818946
0.315385 1.580909
0.340000 1.385597
0.364615 1.220802
0.389231 1.092628
0.413846 0.982765
0.438462 0.891212
0.463077 0.811867
0.487692 0.744728
0.512308 0.689797
0.536923 0.640968
0.561538 0.592140
0.586154 0.555519
0.610769 0.525002
0.635385 0.494484
0.660000 0.470070
0.684615 0.451760
0.709231 0.433449
0.733846 0.415139
0.758462 0.402932
0.783077 0.390725
0.807692 0.384621
0.832308 0.378518
0.856923 0.378518
0.881538 0.378518
0.906154 0.378518
0.930769 0.390725
0.955385 0.409035
0.980000 0.451760
};
  \addlegendentry{$R=0.20$}

  \addplot [thick, mark=*, color=blue] table [x=D, y=P] {%
D P
0.044615 51.281787
0.069231 27.172924
0.093846 17.541586
0.118462 12.640468
0.143077 9.607024
0.167692 7.519624
0.192308 6.121920
0.216923 5.108737
0.241538 4.303074
0.266154 3.698827
0.290769 3.241064
0.315385 2.844335
0.340000 2.533056
0.364615 2.282813
0.389231 2.069190
0.413846 1.886084
0.438462 1.733497
0.463077 1.605323
0.487692 1.501563
0.512308 1.403907
0.536923 1.318458
0.561538 1.245216
0.586154 1.184181
0.610769 1.135353
0.635385 1.092628
0.660000 1.049904
0.684615 1.019386
0.709231 0.988869
0.733846 0.964455
0.758462 0.946144
0.783077 0.933937
0.807692 0.927833
0.832308 0.921730
0.856923 0.921730
0.881538 0.927833
0.906154 0.946144
0.930769 0.970558
0.955385 1.019386
0.980000 1.098732
};
  \addlegendentry{$R=0.40$}

  \addplot [thick, mark=*, color=green] table [x=D, y=P] {%
D P
0.044615 75.402857
0.069231 40.966856
0.093846 26.324537
0.118462 19.043050
0.143077 14.752283
0.167692 11.816494
0.192308 9.637542
0.216923 8.111664
0.241538 6.988618
0.266154 6.079195
0.290769 5.334567
0.315385 4.754734
0.340000 4.290867
0.364615 3.894139
0.389231 3.552343
0.413846 3.271581
0.438462 3.039648
0.463077 2.838232
0.487692 2.655127
0.512308 2.502539
0.536923 2.368262
0.561538 2.258399
0.586154 2.160742
0.610769 2.075293
0.635385 2.002051
0.660000 1.934913
0.684615 1.879981
0.709231 1.831153
0.733846 1.794532
0.758462 1.764014
0.783077 1.745704
0.807692 1.733497
0.832308 1.727393
0.856923 1.727393
0.881538 1.739600
0.906154 1.764014
0.930769 1.800635
0.955385 1.867774
0.980000 1.983741
};
  \addlegendentry{$R=0.60$}

  \end{axis}
\end{tikzpicture}}
\end{center}
\caption{Power-distortion trade-off for various rates. \label{fig:rd1}}
\end{figure}
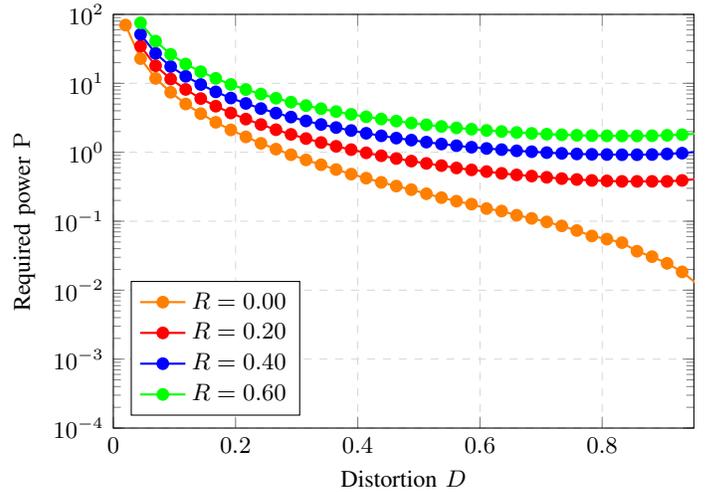

\section{Conclusion}\label{sec:concl}
This letter derived the optimal rate–distortion trade-off for a state-dependent fading Gaussian channel with joint communication and common reconstruction requirements. Future work may explore extensions to multi-user fading scenarios and the design of optimal power adaptation strategies across fading.

\bibliographystyle{IEEEtran}
\bibliography{mybib.bib}

\begin{thebibliography}{10}
\providecommand{\url}[1]{#1}
\csname url@samestyle\endcsname
\providecommand{\newblock}{\relax}
\providecommand{\bibinfo}[2]{#2}
\providecommand{\BIBentrySTDinterwordspacing}{\spaceskip=0pt\relax}
\providecommand{\BIBentryALTinterwordstretchfactor}{4}
\providecommand{\BIBentryALTinterwordspacing}{\spaceskip=\fontdimen2\font plus
\BIBentryALTinterwordstretchfactor\fontdimen3\font minus
  \fontdimen4\font\relax}
\providecommand{\BIBforeignlanguage}[2]{{%
\expandafter\ifx\csname l@#1\endcsname\relax
\typeout{** WARNING: IEEEtran.bst: No hyphenation pattern has been}%
\typeout{** loaded for the language `#1'. Using the pattern for}%
\typeout{** the default language instead.}%
\else
\language=\csname l@#1\endcsname
\fi
#2}}
\providecommand{\BIBdecl}{\relax}
\BIBdecl

\bibitem{gel1980coding}
S.~Gelfand and M.~Pinsker, ``Coding for channels with random parameters,''
  \emph{Probl. Contr. Inform. Theory}, vol.~9, no.~1, pp. 19--31, 1980.

\bibitem{costa1983writing}
M.~H. Costa, ``Writing on dirty paper (corresp.),'' \emph{IEEE Trans. Info.
  Theory}, vol.~29, no.~3, pp. 439--441, 1983.

\bibitem{sutivong2005channel}
A.~Sutivong, M.~Chiang, T.~M. Cover, and Y.-H. Kim, ``Channel capacity and
  state estimation for state-dependent {G}aussian channels,'' \emph{IEEE Trans.
  Info. Theory}, vol.~51, no.~4, pp. 1486--1495, 2005.

\bibitem{steinberg2009coding}
Y.~Steinberg, ``Coding and common reconstruction,'' \emph{IEEE Trans. Info.
  Theory}, vol.~55, no.~11, pp. 4995--5010, 2009.

\bibitem{ramachandran2018state}
V.~Ramachandran, M.~Sreenivasan, S.~R.~B. Pillai, and V.~M. Prabhakaran,
  ``State-dependent {G}aussian broadcast channel with common state
  reconstructions,'' in \emph{2018 International Symposium on Information
  Theory and Its Applications (ISITA)}.\hskip 1em plus 0.5em minus 0.4em\relax
  IEEE, 2018, pp. 658--662.

\bibitem{ahmadi2013heegard}
B.~Ahmadi, R.~Tandon, O.~Simeone, and H.~V. Poor, ``Heegard--berger and cascade
  source coding with common reconstruction constraints,'' \emph{IEEE Trans.
  Info. Theory}, vol.~59, no.~3, pp. 1458--1474, 2013.

\bibitem{timo2010lossy}
R.~Timo, A.~Grant, and G.~Kramer, ``Lossy broadcasting with complementary side
  information,'' \emph{IEEE Trans. Info. Theory}, vol.~59, no.~1, pp. 104--131,
  Jan 2013.

\bibitem{vellambi2014successive}
B.~N. Vellambi and R.~Timo, ``Successive refinement with common receiver
  reconstructions,'' in \emph{2014 IEEE International Symposium on Information
  Theory}.\hskip 1em plus 0.5em minus 0.4em\relax IEEE, 2014, pp. 2664--2668.

\bibitem{lu2021vector}
J.~Lu and Y.~Xu, ``Vector gaussian heegard-berger problem with common
  reconstructions,'' in \emph{2021 13th International Conference on Wireless
  Communications and Signal Processing (WCSP)}.\hskip 1em plus 0.5em minus
  0.4em\relax IEEE, 2021, pp. 1--5.

\bibitem{adikari2022two}
T.~Adikari and S.~Draper, ``Two-terminal source coding with common sum
  reconstruction,'' in \emph{2022 IEEE International Symposium on Information
  Theory (ISIT)}.\hskip 1em plus 0.5em minus 0.4em\relax IEEE, 2022, pp.
  1420--1424.

\bibitem{adikari2024common}
T.~Adikari and S.~C. Draper, ``Common function reconstruction with information
  swapping terminals,'' in \emph{2024 IEEE International Symposium on
  Information Theory (ISIT)}.\hskip 1em plus 0.5em minus 0.4em\relax IEEE,
  2024, pp. 1670--1675.

\bibitem{vaze2008dirty}
C.~S. Vaze and M.~K. Varanasi, ``Dirty paper coding for fading channels with
  partial transmitter side information,'' in \emph{2008 42nd Asilomar
  Conference on Signals, Systems and Computers}, 2008, pp. 341--345.

\bibitem{ramachandran2021joint}
V.~Ramachandran, ``Joint communication and state estimation over a
  state-dependent fading {G}aussian channel,'' \emph{IEEE Wireless
  Communications Letters}, vol.~11, no.~2, pp. 367--370, 2022.

\bibitem{ramachandran2018communication}
V.~Ramachandran, S.~R.~B. Pillai, and V.~M. Prabhakaran, ``Communication and
  state estimation over a state-dependent {G}aussian multiple-access channel,''
  in \emph{2018 Twenty Fourth National Conference on Communications
  (NCC)}.\hskip 1em plus 0.5em minus 0.4em\relax IEEE, 2017, pp. 1--6.

\bibitem{ramachandran2022joint}
V.~Ramachandran, ``Joint communication and state estimation over a
  state-dependent fading {G}aussian channel,'' \emph{IEEE Wireless
  Communications Letters}, vol.~11, no.~2, pp. 367--370, 2022.

\end{thebibliography}
\nocite{ramachandran2018communication}
\nocite{ramachandran2022joint}
\end{document}